\documentclass[11pt]{article}

\usepackage{amsmath,amssymb,amsfonts,amsthm}
\usepackage{mathtools}
\usepackage{dsfont}
\usepackage{booktabs}
\usepackage{graphicx}
\usepackage{xcolor}
\usepackage{microtype}
\usepackage{authblk}
\usepackage{tikz}

\usepackage[backend=biber,style=numeric,sorting=nyt,maxbibnames=99,natbib=true]{biblatex}

\usepackage{hyperref}

\hypersetup{colorlinks=true,linkcolor=blue,citecolor=blue,urlcolor=blue}

\newtheorem{theorem}{Theorem}[section]
\newtheorem{proposition}[theorem]{Proposition}

\theoremstyle{definition}
\newtheorem{definition}[theorem]{Definition}
\theoremstyle{remark}
\newtheorem{remark}[theorem]{Remark}

\newcommand{\V}{V}
\newcommand{\E}{E}
\newcommand{\R}{\mathbb{R}}
\newcommand{\onevec}{\mathds{1}}
\newcommand{\A}{\mathbf{A}}
\newcommand{\D}{\mathbf{D}}
\newcommand{\W}{\mathbf{W}}
\newcommand{\Tt}{\mathbf{T}}
\newcommand{\Oo}{\mathbf{O}}
\newcommand{\J}{\mathbf{J}}
\newcommand{\diag}{\operatorname{diag}}
\newcommand{\Tr}{\operatorname{Tr}}
\newcommand{\supp}{\operatorname{supp}}
\newcommand{\abs}[1]{\left|#1\right|}

\newcommand{\Id}{I_n}

\title{Two--Path Operators, Triadic Decompositions, and Majorized Quotients for Ego--Centered Network Compression}
\author{Moses Boudourides\\
School of Professional Studies, Northwestern University\\
Evanston, IL, United States\\
\texttt{Moses.Boudourides@northwestern.edu}}
\date{}

\begin{document}
\maketitle

\begin{abstract}
Two--paths (wedges) are the elementary combinatorial objects behind clustering, triadic closure, redundancy, and brokerage.
Motivated by a two--path formalism that links Burt's structural holes to node--centered ego networks, we develop an operator viewpoint
in which wedge incidence induces a canonical ``two--walk'' matrix and a unique decomposition into an edge--supported (triadic) part
and a nonedge--supported (open) part.
We then study quotient/contraction constructions designed to compress collections of dominating ego networks together with selected
``traversing'' nodes, and we prove a two--walk transfer theorem under contraction, establishing an inequality with an explicit nonnegative
error term and an equality characterization in terms of a wedge--equitable partition.
Finally, we illustrate the theory on ten benchmark graphs and their ego--traversing contractions using table--driven diagnostics. 
\end{abstract}

\section{Why wedges as operators matter in network science}
Triadic closure and its complement---open two--paths---are central to network science.
Watts--Strogatz clustering and Newman's transitivity quantify closure at node and graph scales \parencite{watts_strogatz_1998,newman_2003_structure,newman_2010_networks},
while Burt's theory of structural holes interprets openness as redundancy/brokerage in ego networks \parencite{burt_1995_structuralholes,borgatti_1997_structuralholes,perry_pescosolido_borgatti_2018}.
All of these notions are naturally expressed in terms of \emph{two--paths} (wedges), but typical summaries compress the wedge structure into a scalar.
Our goal is to keep wedge information \emph{matrix--valued} and thus suitable for algebraic and spectral analysis, and to connect the resulting invariants
to compression procedures that aggregate ego--centered subgraphs.

The operator viewpoint has two practical payoffs.
First, it yields a canonical decomposition of the two--walk operator into a triadic part supported on edges and an open part supported on nonedges.
Second, it provides a principled way to diagnose how much two--walk structure is distorted by contracting families of ego networks into ``supernodes'',
a common step in network visualization and multiscale modeling.
The main technical point is that naive quotient formulas can fail unless one states explicit regularity hypotheses: contracting vertices into blocks may artificially mix two--walk contributions coming from different vertices in the same block, thereby overcounting the amount of traversable two--step structure in the quotient network.
Consequently, a coarse-grained network can appear to preserve two--walk connectivity even when the original network does not.
We therefore emphasize inequalities with explicit error terms and characterize precisely when equalities hold.

\section{Preliminaries and notation}
Let $G=(\V,\E)$ be a finite, simple, undirected graph with $\V=\{1,\dots,n\}$ and $\abs{\E}=m$.
Write $i\sim j$ when $\{i,j\}\in \E$.
Let $\A\in\{0,1\}^{n\times n}$ be the adjacency matrix and $\D=\diag(d_1,\dots,d_n)$ the degree matrix.

\begin{definition}[Two--paths and wedges]
A \emph{two--path} is an ordered triple $(i,j,k)$ of distinct vertices with $i\sim j$ and $j\sim k$.
It is \emph{closed} if $i\sim k$ and \emph{open} otherwise.
\end{definition}

Let $\tau(i)$ denote the number of triangles containing vertex $i$; the total number of triangles is
$\tau = \frac16\Tr(\A^3)$ and $\sum_{i\in\V}\tau(i)=3\tau$ \parencite{newman_2010_networks}.
Let $\pi_2(i)=\binom{d_i}{2}$ be the number of unordered neighbor pairs of $i$.
The (Watts--Strogatz) local clustering coefficient is $C(i)=\tau(i)/\pi_2(i)$ for $d_i\ge 2$ \parencite{watts_strogatz_1998,serrano_boguna_2006_clustering}.

Define the total number of two--paths in $G$ (counting wedges by middle vertex) as
\begin{equation}\label{eq:m2}
m_2 := \sum_{i\in\V}\binom{d_i}{2} = \frac12\sum_{i\in\V} d_i(d_i-1).
\end{equation}
Newman's transitivity (global clustering coefficient) can be written as $3\tau/m_2$ when $m_2>0$ \parencite{newman_2003_structure}.

\section{Incidence and the two--walk operator}
We recall the wedge incidence construction from \parencite{boudourides_2022_note} and develop an operator normalization.

\subsection{Two--incidence and Gram identities}

Let $B\in\{0,1\}^{n\times m}$ be the (unoriented) incidence matrix of $G$, so that $BB^\top=\D+\A$. This identity motivates an analogous incidence-based representation for two--paths (wedges), which we now formalize. To this end, index the set of (unordered) two--paths by $\{1,\dots,m_2\}$.

\begin{definition}[Two--incidence matrix]
The \emph{two--incidence matrix} $B_2\in\{0,1\}^{n\times m_2}$ is defined by
\[
(B_2)_{i,p} =
\begin{cases}
1, & \text{if vertex $i$ is an endpoint of the two--path $p$},\\
0, & \text{otherwise.}
\end{cases}
\]
\end{definition}

\begin{proposition}[Wedge Gram identity]\label{prop:B2}
There exists a diagonal matrix $\D_2=\diag(d_{2,1},\dots,d_{2,n})$ such that
\begin{equation}\label{eq:B2Gram}
B_2B_2^\top = \D_2 + \A^2.
\end{equation}
Moreover, $d_{2,i}=\sum_{j}( \A^2)_{ij}-d_i$ is the number of two--paths incident to $i$ (counting by endpoints), and
$\sum_i d_{2,i}=2m_2$.
\end{proposition}

\begin{proof}
Fix $i\neq j$. The $(i,j)$ entry of $B_2B_2^\top$ counts the number of two--paths for which both $i$ and $j$ are endpoints.
Such a two--path exists precisely when $i$ and $j$ have a common neighbor $k$, and each common neighbor contributes exactly one such two--path.
Hence $(B_2B_2^\top)_{ij}=(\A^2)_{ij}$ for $i\neq j$.
The diagonal entry $(B_2B_2^\top)_{ii}$ counts two--paths incident to $i$ as an endpoint; denote this by $d_{2,i}$ and set $\D_2=\diag(d_{2,i})$.
Finally, summing $d_{2,i}$ over $i$ counts each two--path twice (once per endpoint), giving $\sum_i d_{2,i}=2m_2.$
\end{proof}

\subsection{The canonical two--walk operator}

\begin{definition}[Two--walk operator]\label{def:W}
The \emph{two--walk operator} (endpoint co--incidence) is the matrix $\W\in\mathbb{R}^{n\times n}$ defined by
\begin{equation}\label{eq:Wdef}
\W := \A^2 - \D.
\end{equation}
For $i\neq j$, $\W_{ij}$ equals the number of common neighbors of $i$ and $j$.
The diagonal convention in \eqref{eq:Wdef} is chosen so that $\W\onevec = (d_{2,1},\dots,d_{2,n})^\top$.
Proposition~\ref{prop:B2} implies $\W = B_2B_2^\top -\D_2-\D$.
\end{definition}

\section{A unique triadic/open decomposition of the wedge operator}
We now express triadic closure and openness as a unique support--restricted decomposition of $\W$.

\subsection{Edge triangle multiplicity and triadic intensity}

\begin{definition}[Triangle multiplicity]\label{def:triangle_multiplicity}
For an edge $e=\{i,j\}\in\E$, the \emph{triangle multiplicity} $t_e$ is the number of triangles containing $e$.
Equivalently, for $\{i,j\}\in\E$,
\[
t_{\{i,j\}}=(\A^2)_{ij}.
\]
In particular, triangle multiplicity is the restriction of the two--walk operator to edges, i.e., $t_{\{i,j\}}=\W_{ij}$ whenever $\{i,j\}\in\E$.
\end{definition}

\begin{definition}[Triadic and open parts of $\W$]\label{def:TO}
Let $\J=\onevec\onevec^\top$.
Define matrices
\begin{equation}\label{eq:TOdef}
\Tt := \A\circ \W,\qquad
\Oo := (\J-\Id-\A)\circ \W,
\end{equation}
where $\circ$ denotes the Hadamard (entrywise) product.
Here $\Id$ denotes the $n\times n$ identity matrix, so that $\J-\Id-\A$ masks nonedges off the diagonal.
\end{definition}

\begin{theorem}[Canonical wedge decomposition and uniqueness]\label{thm:decomp}
Let $G$ be any simple undirected graph and define $\W$ by \eqref{eq:Wdef}.
Then:
\begin{enumerate}
\item For all $i\neq j$, $\W_{ij}=\Tt_{ij}+\Oo_{ij}$ and $\Tt_{ij}\Oo_{ij}=0$.
\item $\supp(\Tt)\subseteq \E$ (off-diagonal support on edges) and $\supp(\Oo)\subseteq \V\times\V\setminus(\{(i,i): i\in \V\}\cup \E)$ (off-diagonal support on nonedges).
\item (Uniqueness) If $X,Y\in\R^{n\times n}$ satisfy $\W_{ij}=X_{ij}+Y_{ij}$ for all $i\neq j$, with $X_{ij}=0$ whenever $\{i,j\}\notin\E$ and
$Y_{ij}=0$ whenever $\{i,j\}\in\E$, then $X=\Tt$ and $Y=\Oo$ off-diagonal.
\end{enumerate}
Moreover, for $\{i,j\}\in\E$ we have $\Tt_{ij}=t_{\{i,j\}}$ (triangle multiplicity), and for $\{i,j\}\notin\E$, $\Oo_{ij}=\W_{ij}$ counts open two--paths with endpoints $i,j$.
\end{theorem}

\begin{proof}
For $i\neq j$, either $\{i,j\}\in\E$ or not.
If $\{i,j\}\in\E$, then $(\J-\Id-\A)_{ij}=0$ and $\A_{ij}=1$, hence $\Tt_{ij}=\W_{ij}$ and $\Oo_{ij}=0$.
If $\{i,j\}\notin\E$, then $\A_{ij}=0$ and $(\J-\Id-\A)_{ij}=1$, hence $\Tt_{ij}=0$ and $\Oo_{ij}=\W_{ij}$.
This proves (1)--(2) and the interpretation.
For uniqueness, fix $i\neq j$.
If $\{i,j\}\in\E$, the support constraints force $Y_{ij}=0$, hence $X_{ij}=\W_{ij}=\Tt_{ij}$.
If $\{i,j\}\notin\E$, the support constraints force $X_{ij}=0$, hence $Y_{ij}=\W_{ij}=\Oo_{ij}$.
\end{proof}

\subsection{A sharp extremal characterization via open--wedge mass}

\begin{definition}[Open--wedge mass]\label{def:omega}
The \emph{open--wedge mass} (global openness) of a graph $G$ is the scalar
\begin{equation}\label{eq:omega}
\omega(G) := m_2 - 3\tau,
\end{equation}
where $m_2=\sum_{i\in\V}\binom{d_i}{2}$ is the total number of two--paths (wedges) counted by middle vertex, and $\tau$ is the number of triangles in $G$.
\end{definition}

This definition is natural because it separates total wedge mass into its closed (triadic) and open components. Indeed, using $\sum_e t_e = 3\tau$ (each triangle contributes three closed wedges, one per choice of middle vertex), we have
\[
m_2 = 3\tau + (\text{number of open wedges}),
\]
so that $\omega(G)=m_2-3\tau$ is exactly the total number of open wedges counted by middle vertex.
In particular, this quantity admits an equivalent expression in terms of the open component $\Oo$ in the decomposition $\W=\Tt+\Oo$ (Theorem~\ref{thm:decomp}), as shown below.
In particular, this quantity admits an equivalent expression in terms of the open component $\Oo$ in the decomposition $\W=\Tt+\Oo$, as shown below.

\begin{definition}[$P_3$-free graph]
A graph is called \emph{$P_3$-free} if it contains no induced path on three vertices.
Equivalently, every connected component of a $P_3$-free graph is a clique.
\end{definition}

\begin{theorem}[Sharp openness inequality and equality graphs]\label{thm:clustergraph}
For every simple graph $G$, $\omega(G)\ge 0$.
Moreover, $\omega(G)=0$ if and only if $G$ is a disjoint union of cliques (equivalently, $G$ is $P_3$--free).
\end{theorem}

\begin{proof}

By definition, every two--path (wedge) is either open or closed. Counting wedges by middle vertex gives $m_2$ total wedges, and each triangle contributes exactly three closed wedges (one for each choice of middle vertex). Hence
\[
m_2 = 3\tau + (\text{number of open wedges}),
\]
so that $\omega(G)=m_2-3\tau$ equals the total number of open wedges, and in particular $\omega(G)\ge 0$.

If $G$ is a disjoint union of cliques, then every wedge is closed, so $\omega(G)=0$.
Conversely, if $\omega(G)=0$ then there are no open wedges, i.e., no induced path on three vertices.
Thus $G$ is $P_3$--free.
A standard characterization of $P_3$--free graphs is that each connected component is a clique (sometimes called a \emph{cluster graph}); see, e.g., \citep[Ch.~1]{diestel_2025_graphtheory}.
\end{proof}

\section{Open wedge matrix and redundancy as nonedge common neighbors}\label{sec:openwedge}
The open part $\Oo$ from Definition~\ref{def:TO} records \emph{nonedge} pairs that share neighbors.
This gives a direct bridge to structural holes: nonadjacent alters in an ego neighborhood create open wedges, and $\Oo$ aggregates these across the whole graph.

\begin{theorem}[Open wedge mass as a nonedge sum]\label{thm:omega_sumO}
Let $\Oo$ be defined by \eqref{eq:TOdef}. Then the open--wedge mass admits the equivalent expression in terms of nonedge entries of $\W$:
\begin{equation}\label{eq:omega_sumO}
\omega(G)=m_2-3\tau = \sum_{\substack{i<j\\ i\not\sim j}} (\A^2)_{ij}
= \frac12\sum_{\substack{i\neq j\\ i\not\sim j}} \W_{ij}.
\end{equation}
\end{theorem}

\begin{proof}
For a fixed unordered pair $\{i,j\}$ with $i\not\sim j$, the entry $\W_{ij}=(\A^2)_{ij}$ counts the number of common neighbors $k$ of $i$ and $j$.
Each such $k$ yields an open two--path $(i,k,j)$ whose endpoints are precisely $i$ and $j$.
Thus $\sum_{i<j,\, i\not\sim j}\W_{ij}$ counts each open wedge exactly once.
Since $(\A^2)_{ij}=\W_{ij}$ for $i\neq j$, the first equality follows.
The second equality is just symmetrization over $i\neq j$.
Finally, $\omega(G)$ equals the total number of open wedges (counted by middle vertex), while the sum above counts the same wedges by their endpoints; these are equivalent enumerations of the same set of open wedges, so the identity holds.
\end{proof}

\begin{remark}[Ego redundancy as a row--sum]
For a vertex $v$, the number of open wedges \emph{through} $v$ is $\binom{d_v}{2}-\tau(v)$.
Theorem~\ref{thm:omega_sumO} is the global version of the same idea: $\Oo$ aggregates open wedges by \emph{endpoints} rather than by middle vertex.
\end{remark}

\section{Triangle incidence, edge multiplicities, and operator factorizations}
The decomposition of Theorem~\ref{thm:decomp} isolates the \emph{support} of closure versus openness.
To study closure more finely, we pass from vertex wedges to \emph{edge triangle multiplicities} and an incidence factorization.

\subsection{Edge--triangle multiplicity matrix}

\begin{definition}[Edge triangle multiplicity matrix]
The \emph{edge triangle multiplicity matrix} $M_\triangle\in\R^{n\times n}$ is defined by
\[
(M_\triangle)_{ij} :=
\begin{cases}
(\A^2)_{ij}, & \text{if } i\sim j,\\
0, & \text{otherwise.}
\end{cases}
\]
\end{definition}

Thus $M_\triangle=\Tt$ off-diagonal, and its nonzero entries are exactly the triangle multiplicities $t_{\{i,j\}}$.

\begin{remark}
The matrix $M_\triangle$ coincides with the triadic matrix $\Tt$ defined in
Definition~\ref{def:TO}. Indeed, for $\{i,j\}\in\E$ we have
$(M_\triangle)_{ij}=(\A^2)_{ij}=t_{\{i,j\}}=\Tt_{ij}$, while for
$\{i,j\}\notin\E$ both matrices vanish, and both have zero diagonal.
Thus $M_\triangle=\Tt$. We retain the notation $M_\triangle$ to emphasize its
interpretation as the edge triangle multiplicity matrix, isolating the
triangle-supported part of the two-walk operator.
\end{remark}

A basic identity connects $M_\triangle$ to vertex triangle counts.

\begin{proposition}[Edge multiplicities sum to triangles]\label{prop:edgesum}
The following identities hold: the first globally over all edges, and the second for each vertex $i\in\V$:
\begin{equation}\label{eq:edgesum}
\sum_{\{i,j\}\in\E} t_{\{i,j\}} = 3\tau,
\end{equation}
and for every vertex $i\in\V$,
\begin{equation}
\sum_{j\sim i} t_{\{i,j\}} = 2\tau(i).
\end{equation}
\end{proposition}

\begin{proof}
Each triangle has exactly three edges, so it contributes $1$ to the multiplicity of each of its three edges, giving the first identity.
Fixing a vertex $i$, each triangle containing $i$ has exactly two edges incident to $i$, hence contributes $1$ to $t_{\{i,j\}}$ for each of those two edges.
Summing over neighbors $j\sim i$ counts each triangle containing $i$ twice, giving the second identity.
\end{proof}

\subsection{Triangle incidence and a Gram factorization}
Let $\mathcal{T}$ denote the set of triangles in $G$ (unordered $3$--cycles), $\abs{\mathcal{T}}=\tau$.
Define the \emph{triangle incidence matrix} $B_\triangle\in\{0,1\}^{n\times \tau}$ by
$(B_\triangle)_{i,t}=1$ if vertex $i$ lies in triangle $t$ and $0$ otherwise.

\begin{theorem}[Triangle Gram factorization]\label{thm:trianglegram}
The matrix $B_\triangle B_\triangle^\top$ satisfies
\[
(B_\triangle B_\triangle^\top)_{ii}=\tau(i),\qquad
(B_\triangle B_\triangle^\top)_{ij}=
\begin{cases}
t_{\{i,j\}}, & i\sim j,\\
0, & i\not\sim j.
\end{cases}
\]
Equivalently,
\begin{equation}\label{eq:trianglegram}
B_\triangle B_\triangle^\top = \diag(\tau(1),\dots,\tau(n)) + M_\triangle.
\end{equation}
\end{theorem}

\begin{proof}
The diagonal entry counts the number of triangles containing vertex $i$, which is $\tau(i)$ by definition.
For $i\neq j$, the $(i,j)$ entry counts triangles containing both vertices $i$ and $j$.
Such a triangle exists if and only if $i\sim j$ and the edge $\{i,j\}$ belongs to that triangle; the count is precisely $t_{\{i,j\}}$ when $i\sim j$, and $0$ otherwise.
\end{proof}

\begin{remark}[Why this is more than clustering]
$B_\triangle$ and \eqref{eq:trianglegram} retain the full distribution of triangle overlap patterns: two vertices can be strongly coupled by many shared triangles even if their local clustering coefficients are modest.
This allows one to study closure using linear algebra (spectra, norms, ranks) rather than scalars alone.
\end{remark}

\section{Spectral bounds and extremal structure}
We record two bounds that connect triadicity to classical spectral graph theory.
They are not ``new'' in isolation, but they are essential for a proof--oriented paper because they provide rigorous control of the operators we introduce.
Throughout, let $\lambda_1\ge \cdots\ge \lambda_n$ be the eigenvalues of $\A$.

\begin{proposition}[A universal upper bound on triangles]\label{prop:tri_spectral_upper}
For every graph $G$,
\begin{equation}\label{eq:tri_upper}
\tau = \frac{1}{6}\Tr(\A^3) = \frac{1}{6}\sum_{i=1}^n \lambda_i^3 \le \frac{\lambda_1}{6}\sum_{i=1}^n \lambda_i^2 = \frac{\lambda_1}{3}m.
\end{equation}
Equality holds if and only if all nonprincipal eigenvalues satisfy $\lambda_i^3=\lambda_1\lambda_i^2$, i.e., $\lambda_i\in\{0,\lambda_1\}$ for $i\ge 2$.
\end{proposition}

\begin{proof}
The trace identities are standard \parencite{brouwer_haemers_2012,newman_2010_networks}.
Since $\lambda_1\ge \lambda_i$ for all $i$, we have $\lambda_i^3\le \lambda_1\lambda_i^2$ for each $i$.
Summing yields $\sum_i\lambda_i^3\le \lambda_1\sum_i\lambda_i^2$.
Finally, $\sum_i\lambda_i^2=\Tr(\A^2)=2m$.
The equality characterization follows from entrywise equality in $\lambda_i^3\le \lambda_1\lambda_i^2$.
\end{proof}

\begin{proposition}[An operator--norm bound for closure mass]\label{prop:normbound}
Let $\|\cdot\|_F$ and $\|\cdot\|_2$ denote the Frobenius and spectral norms.
Then
\[
\|\Tt\|_F^2 = \sum_{\{i,j\}\in\E} t_{\{i,j\}}^2 \le \|\A^2\|_F^2 = \sum_{i=1}^n \lambda_i^4 \le n\lambda_1^4.
\]
Here $\|\A^2\|_F^2=\operatorname{Tr}(\A^4)=\sum_{i=1}^n \lambda_i^4$ since $\A^2$ is symmetric.
In particular, closure concentrated on a small set of edges forces large $\|\Tt\|_F$ and therefore a large spectral radius.
\end{proposition}

\begin{proof}
The first identity is the definition of $\Tt$ as edge--supported entries equal to $t_{\{i,j\}}$.
Entrywise, $\Tt$ is obtained from $\A^2$ by zeroing nonedges, hence $\|\Tt\|_F\le \|\A^2\|_F$.
Finally, $\|\A^2\|_F^2=\Tr(\A^4)=\sum_i\lambda_i^4\le n\lambda_1^4$.
\end{proof}

\begin{remark}[Extremal intuition]
Theorem~\ref{thm:clustergraph} shows that \emph{all} wedges are closed exactly for cluster graphs.
Proposition~\ref{prop:tri_spectral_upper} complements this by controlling the total number of triangles by $(m,\lambda_1)$, which is useful when comparing graphs of different sizes.
\end{remark}

\section{Clustered and traversing vertices}
The note \parencite{boudourides_2022_note} introduces \emph{clustered} vertices (those in at least one triangle) and \emph{traversing} vertices (triangle--free),
and discusses bridges.
The bridge statement must be handled carefully: triangle--free does \emph{not} imply all incident edges are bridges.
The correct local characterization is purely neighborhood--based.

\begin{definition}[Clustered and traversing vertices]
A vertex $v$ is \emph{clustered} if $\tau(v)>0$, and \emph{traversing} if $\tau(v)=0$.
Write $\V_{\mathrm{cl}}$ and $\V_{\mathrm{tr}}$ for the corresponding sets.
\end{definition}

\begin{proposition}[Equivalent traversing characterizations]\label{prop:traversing}
For a vertex $v$, the following are equivalent:
\begin{enumerate}
\item $v$ is traversing, i.e., $\tau(v)=0$.
\item Let $N(v)=\{u\in\V : u\sim v\}$ denote the (open) neighborhood of $v$.
The induced subgraph $G[N(v)]$ has no edges.
\item For every pair of distinct neighbors $x,y\in N(v)$, we have $x\not\sim y$.
\end{enumerate}
\end{proposition}

\begin{proof}
$v$ lies in a triangle if and only if two of its neighbors are adjacent.
Thus $\tau(v)=0$ if and only if no edge exists among vertices in $N(v)$, proving the equivalence.
\end{proof}

\section{Combinatorial bounds for clustered structure and ego domination}\label{sec:bounds}
We collect two elementary but useful bounds that turn qualitative statements about ``clustered'' vertices into quantitative lemmas
that can be used later in proofs.

\begin{proposition}[Clustered vertices versus triangle count]\label{prop:Vc_tau}
If $\tau>0$ then $3\le \abs{\V_{\mathrm{cl}}}\le 3\tau$.
\end{proposition}

\begin{proof}
If $\tau>0$, at least one triangle exists and contributes three clustered vertices, so $\abs{\V_{\mathrm{cl}}}\ge 3$.
Conversely, every clustered vertex lies in at least one triangle.
Counting vertex--triangle incidences gives $\sum_{v\in\V_{\mathrm{cl}}}\tau(v)=3\tau$ and each $\tau(v)\ge 1$, hence $\abs{\V_{\mathrm{cl}}}\le 3\tau$.
\end{proof}

\begin{definition}[Traversing vertex types]\label{def:Ttypes}
Let $G$ be decomposed into clustered vertices $\V_{\mathrm{cl}}$ and traversing vertices $\V_{\mathrm{tr}}$.
Among the traversing vertices we distinguish two types:
\begin{itemize}
\item $T_3$: traversing vertices adjacent to exactly one clustered block but not internal to it;
\item $T_4$: traversing vertices adjacent to two or more clustered blocks.
\end{itemize}
Thus $\V_{\mathrm{tr}} = T_3 \cup T_4$.
\end{definition}

\begin{proposition}[Dominating set size and quotient size]\label{prop:quot_size}
Let $S$ be a dominating set of the clustered core $H=G[\V_{\mathrm{cl}}]$ and construct the ego--traversing partition $\Pi$ by aggregating
dominating ego blocks and retaining $T_3\cup T_4$ traversing vertices as singleton blocks.
Then the number of blocks is
\[
r = \abs{S} + \abs{T_3} + \abs{T_4}.
\]
\end{proposition}

\begin{proof}
By construction there is exactly one block per dominating ego center $s\in S$.
All other vertices are assigned to one of these ego blocks except those traversing vertices retained as singletons, which are precisely $T_3\cup T_4$.
The blocks are disjoint and cover $\V$, so the count follows.
\end{proof}

\begin{remark}
Proposition~\ref{prop:Vc_tau} and the identity $\omega=m_2-3\tau$ imply that $\abs{\V_{\mathrm{cl}}}$ is simultaneously controlled by closure ($\tau$)
and openness ($m_2$).
In particular, for sparse graphs with fixed $m$ and small $\lambda_1$ (Proposition~\ref{prop:tri_spectral_upper}), $\tau$ and thus $\abs{\V_{\mathrm{cl}}}$ cannot be large.
\end{remark}

\section{Dominating ego networks and a majorized contraction theorem}
We now formalize the ``dominating ego networks + traversing nodes'' construction from \parencite{boudourides_2022_note}
and prove a two--walk transfer theorem under contraction.

\subsection{Dominating sets on the clustered core}
Assume $\V_{\mathrm{cl}}\neq\emptyset$ and let $H:=G[\V_{\mathrm{cl}}]$ be the induced subgraph on clustered vertices.
A set $S\subseteq \V_{\mathrm{cl}}$ is a \emph{dominating set} of $H$ if every vertex in $\V_{\mathrm{cl}}$ is in $S$ or adjacent in $H$ to some $s\in S$ \parencite{haynes_hedetniemi_slater_1998}.
Vertices in $S$ are called \emph{dominating clustered vertices} and those in $\V_{\mathrm{cl}}\setminus S$ \emph{dominated clustered vertices}.

Following \parencite{boudourides_2022_note}, we also partition traversing vertices into four types according to their adjacency to dominating/dominated clustered vertices.
We use this only to define which traversing vertices are retained as singletons in the contraction.

\subsection{Partitions and the directed edge--sum quotient}

Let $\Pi=\{P_1,\dots,P_r\}$ be a partition of $\V$ into disjoint nonempty blocks.
Write $b(v)\in\{1,\dots,r\}$ for the block index of $v$.

\begin{definition}[Directed edge--sum quotient matrix]
Let $\Pi=\{P_1,\dots,P_r\}$ be a partition of $\V$ and let $b(v)$ denote the block index of $v$.
The \emph{directed edge--sum quotient matrix} $B\in\R^{r\times r}$ is defined by
\begin{equation}\label{eq:Bdef}
B_{ab} := \#\{ (u,v)\in \V\times\V : u\sim v,\ b(u)=a,\ b(v)=b\}.
\end{equation}
\end{definition}
Thus, for $a\neq b$, $B_{ab}$ equals the number of edges between $P_a$ and $P_b$ (counted once per direction), and $B_{aa}=2\abs{\E(P_a)}$ counts internal directed edges.
Since the underlying graph is undirected, the matrix $B$ is symmetric off--diagonal, i.e., $B_{ab}=B_{ba}$ for $a\neq b$.

\begin{definition}[Aggregated two--walk matrix]
Let $\Pi=\{P_1,\dots,P_r\}$ be a partition of $\V$ and let $\A$ be the adjacency matrix of $G$.
The \emph{aggregated two--walk matrix} $M\in\R^{r\times r}$ is defined by
\begin{equation}\label{eq:Mdef}
M_{ab} := \sum_{\substack{u\in P_a\\ w\in P_b}} (\A^2)_{uw},
\end{equation}
where $(\A^2)_{uw}$ counts the number of length--2 walks between vertices $u$ and $w$.
\end{definition}

\subsection{Wedge--equitable partitions and comparison to equitable partitions}
The equality regime in next Theorem~\ref{thm:transfer} defines a regularity notion that is strictly stronger than the classical equitability
condition for adjacency quotients \parencite{godsil_royle_2001,brouwer_haemers_2012}.

\begin{definition}[Equitable and wedge--equitable partitions]\label{def:equitable}
A partition $\Pi=\{P_1,\dots,P_r\}$ is \emph{equitable} (for $\A$) if for every pair $(a,c)$, every vertex $x\in P_c$ has the same number of neighbors in $P_a$;
i.e., $d_a(x)$ is constant over $x\in P_c$.
We call $\Pi$ \emph{wedge--equitable} if for every triple of blocks $(a,b,c)$ there exists a vertex $x^\star\in P_c$ such that
$d_a(x)d_b(y)=0$ for all $x,y\in P_c$ with $(x,y)\neq (x^\star,x^\star)$.
\end{definition}

\begin{proposition}\label{prop:wedge_strict}
Every wedge--equitable partition is equitable, but not conversely.
\end{proposition}

\begin{proof}
If $\Pi$ is wedge--equitable, fix blocks $(a,c)$ and take $b=a$.
Then for each $P_c$ there is a vertex $x^\star$ such that $d_a(x)d_a(y)=0$ for all distinct $x,y\in P_c$.
Thus at most one vertex in $P_c$ can have $d_a(\cdot)>0$, and for all other vertices $d_a(\cdot)=0$.
In particular, $d_a(x)$ is constant on $P_c$ if $P_c$ has size $1$, and otherwise it is constant \emph{except possibly at one vertex}.
However, wedge--equitable also requires this to hold for all $b$ simultaneously, forcing the exceptional vertex to be the same across $a$,
and hence forcing constant neighbor counts across $P_c$ for each $a$.
This is exactly equitability.

To see nonconverse, consider a regular partition in which each block has two vertices each connected in the same way to another block.
Such a partition can be equitable while having two distinct vertices in a block that both connect to the same external block, violating the ``single middle vertex'' requirement in \eqref{eq:error}.
An explicit example is given in the Appendix.
\end{proof}

\begin{remark}[Interpretation]
Equitable partitions preserve adjacency quotients exactly.
Wedge--equitable partitions preserve \emph{two--walk} structure (off diagonal) exactly.
Theorem~\ref{thm:transfer} therefore identifies the correct regularity notion for contractions meant to preserve wedge/triadic information.
\end{remark}

\begin{theorem}[Two--walk transfer under contraction]\label{thm:transfer}
For any partition $\Pi$ of $\V$, the matrices $B$ and $M$ defined by \eqref{eq:Bdef}--\eqref{eq:Mdef} satisfy the entrywise inequality
\begin{equation}\label{eq:transferineq}
M \le B^2.
\end{equation}
Moreover, for $a\neq b$ the overcount has the explicit form
\begin{equation}\label{eq:error}
(B^2-M)_{ab} = \sum_{c=1}^r\ \sum_{\substack{x,y\in P_c\\ x\neq y}} d_a(x)\, d_b(y)\ \ge 0,
\end{equation}
where $d_a(x)$ denotes the number of neighbors of $x$ that lie in block $P_a$.
Equality $M=B^2$ holds off-diagonal if and only if for every block $P_c$ and every pair of blocks $(a,b)$,
the sequences $\{d_a(x)\}_{x\in P_c}$ and $\{d_b(x)\}_{x\in P_c}$ are supported on a common single vertex (equivalently, each block is \emph{wedge--equitable}).
\end{theorem}

\begin{proof}
Fix $a,b$.
By definition of $\A^2$,
\[
M_{ab}=\sum_{u\in P_a}\sum_{w\in P_b}\sum_{v\in\V} \A_{uv}\A_{vw}
=\sum_{c=1}^r\ \sum_{v\in P_c}\Big(\sum_{u\in P_a}\A_{uv}\Big)\Big(\sum_{w\in P_b}\A_{vw}\Big).
\]
For $v\in P_c$, define $d_a(v):=\sum_{u\in P_a}\A_{uv}$ and $d_b(v):=\sum_{w\in P_b}\A_{vw}$.
Then
\[
M_{ab}=\sum_{c=1}^r \sum_{v\in P_c} d_a(v)d_b(v).
\]
By definition of the directed edge--sum quotient matrix,
\[
B_{ac}=\sum_{v\in P_c} d_a(v),
\]
since $d_a(v)$ counts the number of neighbors of $v$ in block $P_a$ and the quotient entry aggregates these contributions over all vertices of $P_c$.

On the other hand,
\[
(B^2)_{ab}=\sum_{c=1}^r B_{ac}B_{cb}
=\sum_{c=1}^r\Big(\sum_{v\in P_c} d_a(v)\Big)\Big(\sum_{v\in P_c} d_b(v)\Big)
=\sum_{c=1}^r\sum_{x,y\in P_c} d_a(x)d_b(y).
\]
Subtracting yields \eqref{eq:error}, hence $M\le B^2$ entrywise.
The equality characterization follows by examining the condition
\[
\sum_{x,y\in P_c} d_a(x)d_b(y) = \sum_{v\in P_c} d_a(v)d_b(v).
\]

($\Rightarrow$) If $M=B^2$ off--diagonal, then for every block $P_c$ and every pair $(a,b)$ all cross terms with $x\neq y$ must vanish. Hence whenever $d_a(x)>0$ and $d_b(y)>0$ with $x,y\in P_c$, we must have $x=y$. Thus the sequences $\{d_a(x)\}_{x\in P_c}$ and $\{d_b(x)\}_{x\in P_c}$ are supported on a common single vertex.

($\Leftarrow$) Conversely, if for every block $P_c$ and every pair $(a,b)$ the sequences $\{d_a(x)\}_{x\in P_c}$ and $\{d_b(x)\}_{x\in P_c}$ are supported on a common single vertex, then all cross terms $d_a(x)d_b(y)$ with $x\neq y$ vanish, and therefore
\[
\sum_{x,y\in P_c} d_a(x)d_b(y) = \sum_{v\in P_c} d_a(v)d_b(v).
\]
Summing over $c$ gives $M_{ab}=(B^2)_{ab}$ for $a\neq b$, as claimed.
\end{proof}

\begin{remark}[Why this avoids quotient traps]
Theorem~\ref{thm:transfer} is deliberately an \emph{inequality} with an explicit error and a sharp equality condition.
This is the inequality form of a ``two--walk transfer'' statement: without wedge--equitable hypotheses, an equality claim is generally false.
\end{remark}

\section{Computational illustration on ten graphs and their contractions}
This section illustrates the theory rather than presenting a dataset study.
We report only invariants tied directly to the preceding theorems and one contraction diagnostic tied to Theorem~\ref{thm:transfer}.

\subsection{Graph set and contraction rule}
We consider ten benchmark graphs (Florentine families, Karate club, Erd\H{o}s--R\'enyi, Dolphins, Les Mis\'erables, Football, Jazz, \emph{C. elegans}, USAir, and Network Science coauthorship).
For each graph we compute the clustered set $\V_{\mathrm{cl}}$, choose a dominating set $S$ of $G[\V_{\mathrm{cl}}]$ (as in \parencite{boudourides_2022_note}),
and build a partition $\Pi$ whose blocks aggregate each dominating ego neighborhood together with assigned dominated and adjacent traversing vertices,
while retaining traversing vertices of types $T_3$ and $T_4$ as singleton blocks.
The directed edge--sum quotient matrix $B$ is then computed from $\Pi$.

\subsection{Invariants and openness}
Table~\ref{tab:base} reports $(n,m,\tau,m_2,\omega,\abs{\V_{\mathrm{cl}}},\abs{\V_{\mathrm{tr}}},\abs{S})$.
The openness invariant $\omega=m_2-3\tau$ is precisely the number of open wedges (Theorem~\ref{thm:clustergraph}) and is therefore a direct measure of ``structural openness''
in the sense of wedge structure.

\begin{table}
\caption{Summary invariants for the ten base graphs.}
\label{tab:base}
\begin{tabular}{lrrrrrrrr}
\toprule
Graph & n & m & triangles & m2 & omega & Vc & Vt & dom \\
\midrule
florentine & 15 & 20 & 3 & 47 & 38 & 7 & 8 & 5 \\
karate & 34 & 78 & 45 & 528 & 393 & 32 & 2 & 4 \\
ErdosRenyi & 50 & 129 & 27 & 685 & 604 & 35 & 15 & 11 \\
dolphins & 62 & 159 & 95 & 923 & 638 & 46 & 16 & 14 \\
LesMiserables & 77 & 254 & 467 & 2808 & 1407 & 57 & 20 & 10 \\
football & 115 & 613 & 810 & 5967 & 3537 & 115 & 0 & 12 \\
jazz & 198 & 2742 & 17899 & 103212 & 49515 & 192 & 6 & 13 \\
celegans & 297 & 2148 & 3241 & 53804 & 44081 & 278 & 19 & 16 \\
usair & 332 & 2126 & 12181 & 92189 & 55646 & 272 & 60 & 36 \\
netscience & 379 & 914 & 921 & 6417 & 3654 & 351 & 28 & 55 \\
\bottomrule
\end{tabular}
\end{table}

Figure~\ref{fig:ecdf} shows the empirical cumulative distribution function (ECDF) of local clustering coefficients $C(v)$.
We emphasize that the wedge-operator framework retains more than this scalar summary:
for example, $\Tt$ records triangle multiplicities per edge, while $\Oo$ records open wedge intensity per nonedge.

\begin{figure}[t]
\centering
\includegraphics[width=\textwidth]{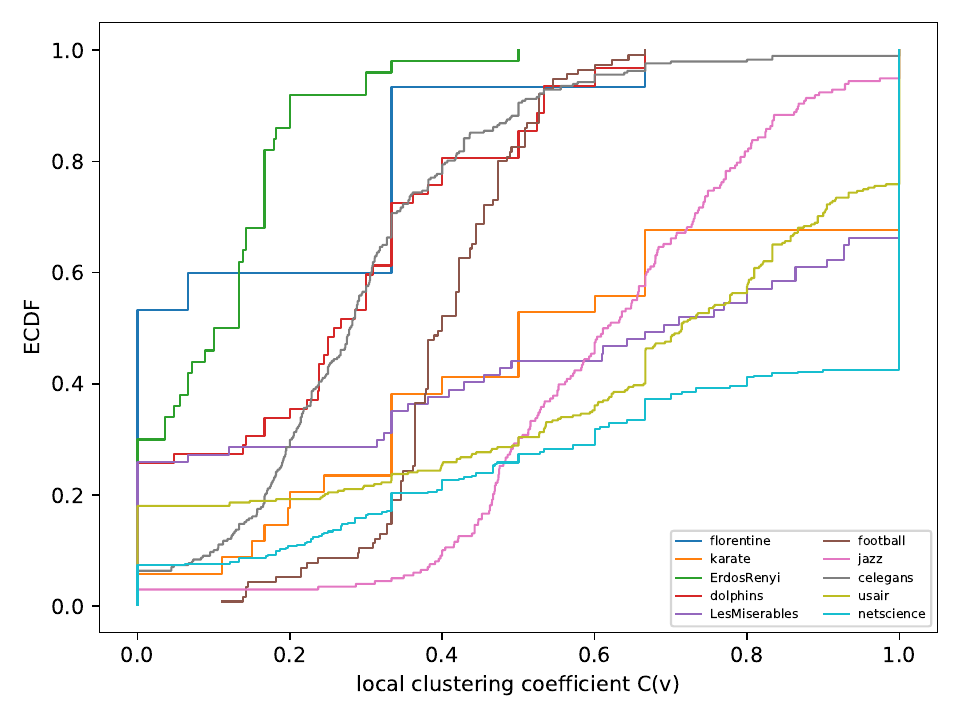}
\caption{ECDF of local clustering coefficients $C(v)$ across vertices, one curve per graph.}
\label{fig:ecdf}

\small\textit{Alt text: Empirical cumulative distribution functions of vertex clustering coefficients for multiple graphs, shown as separate curves, illustrating differences in clustering structure across networks.}

\end{figure}

\subsection{Contraction distortion diagnostic}
Theorem~\ref{thm:transfer} implies that the quotient walk count $B^2$ overcounts aggregated two--walk mass $M$, with a computable nonnegative error.
We report the ratio
\[
\rho := \frac{\sum_{a\neq b} M_{ab}}{\sum_{a\neq b} (B^2)_{ab}} \in (0,1],
\]
where values close to $1$ indicate near wedge--equitable behavior and small distortion.

Table~\ref{tab:quot} reports the two--walk transfer ratio $\rho$ for the ego--traversing partitions of the empirical networks considered here.
The quantity $\rho$ summarizes how closely the contracted quotient preserves two--walk structure: values closer to one indicate little distortion, while smaller values reflect larger overcount due to block aggregation.

\begin{table}
\caption{Quotient/contraction diagnostics for the ego--traversing partition (directed edge-sum quotient).}
\label{tab:quot}
\begin{tabular}{lrrrrrrl}
\toprule
Graph & blocks & egoblocks & TR\_singletons & B\_edges & B\_internal & ratio \\
\midrule
florentine & 8 & 5 & 3 & 9 & 11 & 0.327 \\
karate & 4 & 4 & 0 & 26 & 52 & 0.121 \\
ErdosRenyi & 14 & 11 & 3 & 83 & 46 & 0.174 \\
dolphins & 19 & 14 & 5 & 97 & 62 & 0.188 \\
LesMiserables & 13 & 10 & 3 & 85 & 169 & 0.105 \\
football & 12 & 12 & 0 & 313 & 300 & 0.085 \\
jazz & 18 & 13 & 5 & 1211 & 1531 & 0.021 \\
celegans & 16 & 16 & 0 & 1485 & 663 & 0.040 \\
usair & 47 & 36 & 11 & 976 & 1150 & 0.028 \\
netscience & 60 & 55 & 5 & 242 & 672 & 0.119 \\
\bottomrule
\end{tabular}
\end{table}

\section{Directed and weighted extensions (brief)}
Although we work with simple undirected graphs, the wedge--operator approach extends cleanly.

\subsection{Directed graphs}
For a directed graph with adjacency matrix $\A$, a natural two--step operator from sources to targets is $\A^2$ as before.
However, wedge types can be distinguished: ``out--out'' wedges correspond to $\A^2$, while ``in--in'' wedges correspond to $(\A^\top)^2$,
and mixed wedges correspond to $\A\A^\top$ and $\A^\top\A$.
Each admits a canonical support decomposition analogous to Theorem~\ref{thm:decomp} by Hadamard masking with the directed edge set and its complement.
This yields a family of triadic/open operators tailored to directed transitivity notions in social networks.

\subsection{Weighted graphs}
For a weighted undirected graph with symmetric weight matrix $W=(w_{ij})$, define $\A_W$ with $(\A_W)_{ij}=w_{ij}$ for $i\neq j$ and $0$ on the diagonal.
Then $(\A_W^2)_{ij}=\sum_k w_{ik}w_{kj}$ is a weighted common--neighbor intensity, and one can define
$\W_W=\A_W^2-\diag(\A_W\onevec)$ and its support decomposition exactly as in \eqref{eq:TOdef}.
The contraction theorem (Theorem~\ref{thm:transfer}) also holds verbatim for nonnegative weights, replacing edge counts by weight sums.

\begin{proposition}[Weighted transfer theorem]\label{prop:weighted_transfer}
Let $G$ be a weighted undirected graph with symmetric weights $w_{uv}\ge 0$ and weighted adjacency $\A_W$.
For any partition $\Pi$ define the weighted directed edge--sum quotient $B^{(w)}$ by summing weights in \eqref{eq:Bdef},
and define $M^{(w)}$ by aggregating $(\A_W^2)_{uw}$ as in \eqref{eq:Mdef}.
Then $M^{(w)}\le (B^{(w)})^2$ entrywise, with an error formula identical to \eqref{eq:error} where $d_a(x)$ becomes the total weight from $x$ into block $P_a$.
\end{proposition}

\begin{proof}[Proof sketch]
Repeat the proof of Theorem~\ref{thm:transfer} with $\A$ replaced by $\A_W$.
Nonnegativity of weights ensures that the cross-term expansion remains $\ge 0$.
\end{proof}

\begin{remark}[Hypergraph direction]
Replacing two--paths by length--2 walks in bipartite or hypergraph incidence graphs suggests higher--order analogues of $\W$ and its support decomposition,
bringing the present operator viewpoint closer to motif/higher--order network models \parencite{benson_2016_higherorder}.
\end{remark}

\section{Discussion and directions}
The operator decomposition $\W=\Tt+\Oo$ (Theorem~\ref{thm:decomp}) provides a matrix--valued refinement of clustering that separates triadic closure from openness in a way that is both canonical and algebraically tractable.
The sharp characterization $\omega=0$ (Theorem~\ref{thm:clustergraph}) identifies exactly when wedge openness vanishes: cluster graphs (disjoint unions of cliques).
Finally, the majorized transfer theorem (Theorem~\ref{thm:transfer}) explains how two--walk structure behaves under contraction and makes explicit the wedge--equitable regime in which equality can be expected.

Several mathematically natural extensions are immediate.
First, one may replace $\A^2$ by motif adjacency operators of higher order \parencite{milo_2002_motifs,benson_2016_higherorder} and seek analogous canonical support decompositions.
Second, one can develop spectral bounds for $\tau$ and $\omega$ in terms of $\lambda_{\max}(\A)$ and degree moments, sharpening classical inequalities \parencite{cao_1998_bounds,brouwer_haemers_2012}.
Third, the wedge--equitable condition suggests a new hierarchy of partitions between equitable partitions (for $\A$) and regular partitions for higher--order operators \parencite{godsil_royle_2001}.
We view these as promising directions for a theorem--driven network science agenda.

\section*{Acknowledgements}
This paper develops and extends my 2022 note \parencite{boudourides_2022_note}. I am grateful to E. and G. for their care and support, which made this work possible.

\printbibliography

@misc{boudourides_2022_note,
  author       = {Boudourides, Moses},
  title        = {A Note on 2--Degree Clustering and Dominating Ego Networks in an Undirected Graph},
  howpublished = {manuscript},
  year         = {2022},
  note         = {Draft, January 2022},
  url          = {https://github.com/mboudour/var/blob/master/triadicity.pdf},
}

@article{watts_strogatz_1998,
  title        = {Collective dynamics of `small-world' networks},
  author       = {Watts, Duncan J. and Strogatz, Steven H.},
  journal      = {Nature},
  year         = {1998},
  volume       = {393},
  number       = {6684},
  pages        = {440--442},
  doi          = {10.1038/30918}
}

@article{newman_2003_structure,
  title        = {The Structure and Function of Complex Networks},
  author       = {Newman, Mark E. J.},
  journal      = {SIAM Review},
  year         = {2003},
  volume       = {45},
  number       = {2},
  pages        = {167--256},
  doi          = {10.1137/S003614450342480}
}

@book{newman_2010_networks,
  title        = {Networks: An Introduction},
  author       = {Newman, Mark E. J.},
  publisher    = {Oxford University Press},
  year         = {2010}
}

@article{serrano_boguna_2006_clustering,
  title        = {Clustering in complex networks. I. General formalism},
  author       = {Serrano, M. {\'A}. and Bogu{\~n}{\'a}, M.},
  journal      = {Physical Review E},
  year         = {2006},
  volume       = {74},
  number       = {5},
  pages        = {056114},
  doi          = {10.1103/PhysRevE.74.056114}
}

@book{burt_1995_structuralholes,
  title        = {Structural Holes: The Social Structure of Competition},
  author       = {Burt, Ronald S.},
  publisher    = {Harvard University Press},
  year         = {1995}
}

@article{borgatti_1997_structuralholes,
  title        = {Structural Holes: Unpacking Burt's Redundancy Measures},
  author       = {Borgatti, Stephen P.},
  journal      = {Connections},
  year         = {1997},
  volume       = {20},
  number       = {1},
  pages        = {35--38}
}

@book{perry_pescosolido_borgatti_2018,
  title        = {Egocentric Network Analysis: Foundations, Methods, and Models},
  author       = {Perry, Brea L. and Pescosolido, Bernice A. and Borgatti, Stephen P.},
  publisher    = {Cambridge University Press},
  year         = {2018}
}

@article{cao_1998_bounds,
  title        = {Bounds on eigenvalues and chromatic numbers},
  author       = {Cao, Dasong},
  journal      = {Linear Algebra and its Applications},
  year         = {1998},
  volume       = {270},
  pages        = {1--13},
  doi          = {10.1016/S0024-3795(96)00199-1}
}

@book{diestel_2025_graphtheory,
  title        = {Graph Theory},
  author       = {Diestel, Reinhard},
  edition      = {6},
  publisher    = {Springer},
  year         = {2025},
  doi          = {10.1007/978-3-662-70107-2}
}

@book{haynes_hedetniemi_slater_1998,
  title        = {Fundamentals of Domination in Graphs},
  author       = {Haynes, Teresa W. and Hedetniemi, Stephen T. and Slater, Peter J.},
  publisher    = {Marcel Dekker},
  year         = {1998}
}

@book{godsil_royle_2001,
  title        = {Algebraic Graph Theory},
  author       = {Godsil, Chris and Royle, Gordon},
  publisher    = {Springer},
  year         = {2001}
}

@book{brouwer_haemers_2012,
  title        = {Spectra of Graphs},
  author       = {Brouwer, Andries E. and Haemers, Willem H.},
  publisher    = {Springer},
  year         = {2012}
}

@article{milo_2002_motifs,
  title        = {Network Motifs: Simple Building Blocks of Complex Networks},
  author       = {Milo, Ron and Shen-Orr, Shai and Itzkovitz, Shalev and Kashtan, Nadav and Chklovskii, Dmitri and Alon, Uri},
  journal      = {Science},
  year         = {2002},
  volume       = {298},
  number       = {5594},
  pages        = {824--827},
  doi          = {10.1126/science.298.5594.824}
}

@inproceedings{benson_2016_higherorder,
  title        = {Higher-order organization of complex networks},
  author       = {Benson, Austin R. and Gleich, David F. and Leskovec, Jure},
  booktitle    = {Proceedings of the 22nd ACM SIGKDD International Conference on Knowledge Discovery and Data Mining (KDD)},
  year         = {2016},
  pages        = {115--124},
  url          = {https://arxiv.org/abs/1602.08438}
}

\end{document}